\documentclass[journal]{IEEEtran}
\usepackage[english]{babel}
\usepackage{amsmath,amssymb,amsfonts}
\usepackage{algorithmic}
\usepackage{algorithm}
\usepackage{array}
\usepackage[caption=false,font=normalsize,labelfont=sf,textfont=sf]{subfig}
\usepackage{textcomp}
\usepackage{textcase}
\usepackage{stfloats}
\usepackage{url}
\usepackage{diagbox}
\usepackage{verbatim}
\usepackage{graphicx}
\usepackage{multirow}
\usepackage{cite}
\usepackage{algorithmic}
\usepackage[dvipsnames]{xcolor}
\usepackage{amsthm}

\newtheorem{theoremC}{Corollary}
\newtheorem{theoremP}{Proposition}
\newtheorem{theoremT}{Theorem}
\newtheorem{theoremD}{Definition}

\newtheorem{corollary}[theoremC]{Corollary}
\newtheorem{proposition}[theoremP]{Proposition}
\newtheorem{theorem}[theoremT]{Theorem}
\newtheorem{Definition}[theoremD]{Definition}

\begin{document}
% {Vertex or }
\title{Vertex-Guided Redundant Constraints Identification for Unit Commitment\vspace{5pt}}

\author{Xuan He,  Yuxin Pan, Yize  Chen,~\IEEEmembership{Member, IEEE}, and Danny H.K. Tsang,~\IEEEmembership{Life Fellow, IEEE}\thanks{X. He and D.H.K. Tsang are with Hong Kong University of Science and Technology (Guangzhou); Y. Pan is with Hong Kong University of Science and Technology; Y. Chen is with University of Alberta. Emails: xhe085@connect.hkust-gz.edu.cn, yuxin.pan@connect.ust.hk,  yize.chen@ualberta.ca, eetsang@ust.hk. }}
% % \thanks{This paper was produced by the IEEE Publication Technology Group. They are in Piscataway, NJ.}% <-this % stops a space
% % \thanks{Manuscript received April 19, 2021; revised August 16, 2021.}
% }

\markboth{
In Submission}%
{Shell \MakeLowercase{\textit{et al.}}: 
}


\maketitle
% \textcolor{blue}{Is it necessary to briefly illustrate the drawback of existing methods?} \textcolor{blue}{How about we build a connection between the main contribution you described here and the term ``vertex-guided geometric constraint screening" in the title?}
\begin{abstract}
Power systems Unit Commitment (UC) problem determines the generator commitment schedule and dispatch decisions to realize the reliable and economic operation of power networks. The growing penetration of stochastic renewables and demand behaviors makes it necessary to solve the UC problem timely. It is possible to derive lightweight, faster-to-solve UC models via constraint screening to eliminate redundant constraints. 
However, the screening process remains computationally cumbersome due to the need of solving numerous linear programming (LP) problems. To reduce the number of LPs to solve, we introduce a novel perspective on such classic LP-based screening. Our key insights lie in the principle that redundant constraints will be satisfied by all vertices of the screened feasible region. Using the UC decision variables' bounds tightened by solving much fewer LPs, we build an outer approximation for the UC feasible region as the screened region. A matrix operation is then designed and applied to the outer approximation's vertices to identify all redundant constraints on-the-fly. Adjustments for the outer approximation are further explored to improve screening efficiency by considering the load operating range and cutting planes derived from UC cost and discrete unit status prediction. Extensive simulations are performed on a set of testbeds up to 2,383 buses to substantiate the effectiveness of the proposed schemes. Compared to classic LP-based screening, our schemes can achieve up to 8.8x acceleration while finding the same redundant constraints.
\end{abstract}

% Yuxin:

\begin{IEEEkeywords}
Unit commitment, constraint screening, outer approximation, model reduction.
\end{IEEEkeywords}

% \vspace{-1em}
\section{Introduction}
\IEEEPARstart{U}nit commitment (UC) problem is a crucial task for obtaining reliable and economical operational decisions for power systems~\cite{muralikrishnan2020comprehensive}. UC aims to minimize the operational cost by optimizing the generation schedule involving on/off commitment and dispatch levels. To accommodate grid modernization and renewables integration, UC must be solved quickly and frequently to inform operators and adapt to the dynamic conditions of power systems. Various techniques are developed for solving large-scale UC, such as meta-heuristics \cite{kazarlis2002genetic}, reinforcement learning \cite{de2021applying}, and mixed integer programming (MIP) \cite{chen2022security}. 

Among them, MIP is still the most widely used approach, due to its strong theoretical backing and proven success in real-world applications. With binary variables indicating on/off commitment and continuous variables denoting dispatch levels, generation and other security constraints are also incorporated. Such formulations can be NP-hard and computationally challenging due to their MIP nature~\cite{bendotti2019complexity, Morales2013Tight}. These challenges are further complicated by various network constraints, such as line flow limits, which can substantially slow down the feasibility check and the branch-and-bound procedure of MIP solvers \cite{xavier2019transmission,guo2020fast}. Empirical studies suggest that for a particular demand level, a significant portion of network constraints are redundant, which indeed does not affect the UC decisions \cite{zhai2010fast}. By eliminating these redundant constraints, it is promising to obtain a lightweight UC model that can be solved in a timely manner. This elimination procedure is usually referred as \emph{constraint screening}~\cite{zhai2010fast, roald2019implied, porras2021cost, he2022enabling, he2023fast, telgen1983identifying, zhai2014fast, ding2020fast,awadalla2023tight,weinhold2020fast}.
% and the increasing number of temporal steps~\cite{he2023fast}
% in the scenarios investigated
% that have no impact on the solution optimality.
% The UC decisions of \textcolor{orange}{the resulting lightweight model} will be equivalent to those of the original model. 
\subsection{Motivation}
One type of classic UC screening approaches \cite{zhai2010fast, roald2019implied, porras2021cost, he2022enabling, he2023fast, telgen1983identifying, zhai2014fast, ding2020fast} identifies redundant constraints by comparing the actual bound of the line flow and its predefined bound. This actual bound is given by a linear programming (LP) model as follows,
\begin{subequations} \label{screening1}
\begin{align}
\max_{\boldsymbol{y}}~or \;\min _{\boldsymbol{y}} & \quad \text{power flow}~f_j~\text{through line}~j\label{scre_obj}\\
\text { s.t. } &\quad \boldsymbol{y} \in \Tilde{\boldsymbol{R}}_{uc}
\end{align}
\end{subequations}
where $\boldsymbol{y}$ is the decision variable such as generation dispatch, and $\Tilde{\boldsymbol{R}}_{uc}$ defines a relaxed solution space that contains the original feasible region of the UC model. \cite{zhai2010fast} shows that for model \eqref{screening1} if $f^{*}_j \leq \overline{f}_j$ (or $f^{*}_j \geq \underline{f}_j$) always holds, then the line limit $f_j \leq  \overline{f}_j$ (or $f_j \geq  \underline{f}_j$) will be redundant and removed from the UC model.
% involving the binary unit status and continuous generation dispatch
% Such data-driven models use machine learning (ML) predictions to identify whether constraints are redundant directly 
% By relaxing the binary variables to $[0,1]$ in the original UC model, we can get a relaxed region containing all the original feasible solutions.
\begin{figure*}[t] 
\vspace{-1.5em}
\hspace{1cm}
\includegraphics[scale=0.65]{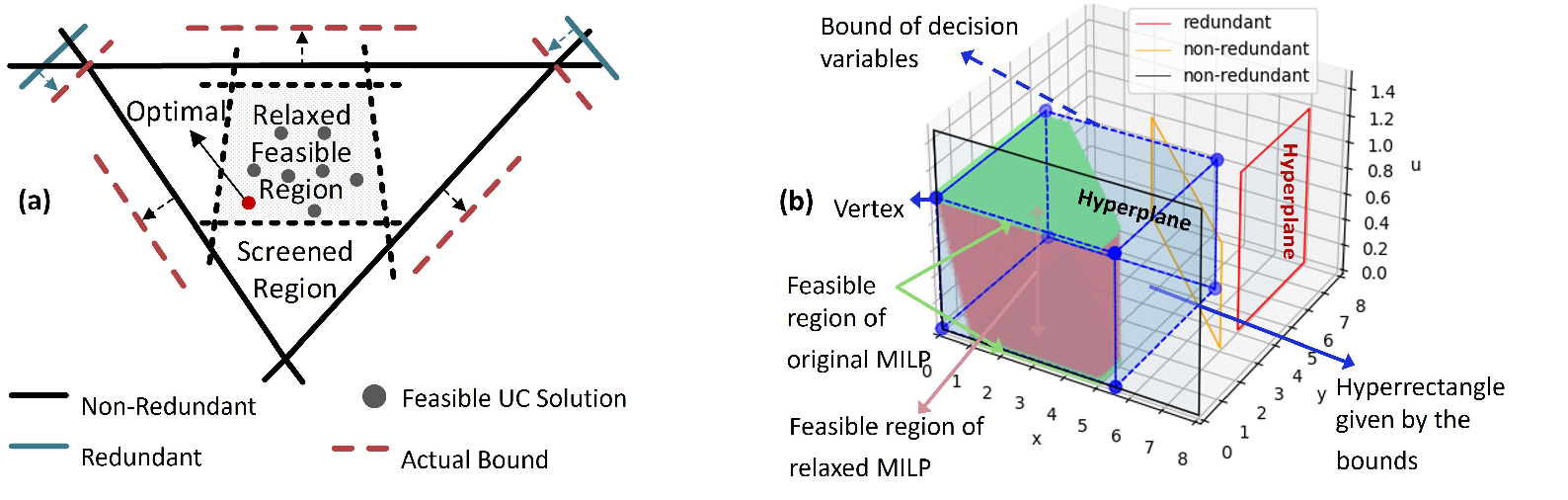}
    \vspace{-0.5em}
	\caption{\footnotesize \textbf{(a) Line flow-guided Screening.} By relaxing binary variables $u$ to $[0,1]$ in the original UC model, we can get a relaxed region containing the UC feasible region. The feasible region of the screening model \eqref{screening1}, i.e., the screened region will be formulated to contain the relaxed region. An actual bound corresponding to the checked limit can be determined by the objective value of \eqref{screening1}. The actual bound is proved to be tighter than the hyperplane representing the redundant constraint. \textbf{(b) Vertex-guided Screening.} The green region $A$ represents the original feasible region of a MILP problem. The pink region $B$ involving $A$ is the feasible region with the relaxed binary variable $u \in [0,1]$. We optimize the upper bounds and lower bounds of the decision variables to obtain the blue hyperrectangle $C$ containing $A$ and $B$. The black, orange, and red rectangles represent the semispace defining the corresponding inequalities. We prove this observation in Theorem 1: if all the vertices of the hyperrectangle satisfy a specific inequality, then this inequality will be redundant for $A$.} \label{Fig-intro-constraint}
% The non-redundant constraints can be hit by the solution in the screened region while may not be hit by the solution in the region A. Beside, will not be hit by any solution in the screened region, either before or after their removal in the screened region. Further,. 
\vspace{-1em}
\end{figure*}
% \textcolor{red}{Overall idea is converyed well in these two paragraphs for classic methods. But too little details are given to the proposed method.}
However, such an approach only utilizes a one-to-one position relationship between the predefined bound and the actual bound of line flow as shown in Fig. \ref{Fig-intro-constraint}. Thus, \emph{it requires solving an LP model for every line limit in the power grid}, which binds the number of LPs and the screening time to the total number of limits. For large-scale systems, such line flow-guided screening (LFGS) becomes cumbersome.

This inspires efforts to shift the anchor of the screening procedure from the line limit itself to other perspectives, such as vertex \cite{awadalla2023tight} and interior point \cite{weinhold2020fast}. Though these efforts can bypass solving an optimization model for each line limit, they either introduce numerous binary variables \cite{awadalla2023tight} or rely on heuristics that do not take advantage of the strong theoretical foundation of LP-based screening~\cite{weinhold2020fast}. Thus, in this paper, we aim to address the following research question:

\emph{What anchor can enable fast UC screening without traversing all line flow limits?}
% \cite{szedlak2017redundancy} 
% enabling the screening time depend mainly on the number of non-redundant constraints. However, 
% \textcolor{blue}{Its better to cite the paper as ``[Authors Name] et al. applies ... [27]"}.
% \textcolor{red}{Too oral}
% \textcolor{blue}{It is better to connect your description for the motivation with the figure to exemplify it} 
% \textcolor{blue}{do we need to explain what is outer approximation}.
% The target screening procedure is supposed to solve fewer optimization models while reducing comparable UC redundancy. 
% Differing from \cite{awadalla2023tight} to capture the binding constraints the vertices lie on,

\subsection{Related Work}
The classic LP-based screening is popular for its simplicity and effectiveness, motivating continued improvements. \cite{roald2019implied, porras2021cost, he2022enabling, he2023fast} convert the load parameters to decision variables limited by a predefined operating range in the screening model \eqref{screening1}. The optimized actual bounds and the binding statuses of line limits thus apply across this range, avoiding repetitively solving \eqref{screening1} for each load input. Additional constraints based on domain knowledge such as the upper bound of the UC cost \cite{porras2021cost,he2022enabling} and the fixed unit on/off statuses \cite{he2023fast} are added as cutting planes to the screening model, which helps tighten the screened region implicitly.

Recent screening paradigms regarding geometric observations \cite{awadalla2023tight, weinhold2020fast}, operation transformation \cite{9314247,hua2013eliminating}, and data-driven methods \cite{zhang2019data, pineda2020data,ardakani2018prediction, xavier2021learning, Cordero2022Warm-starting,sterzinger2023learning} have also attracted attention. \cite{awadalla2023tight} observes that a vertex can capture multiple non-redundant constraints, while \cite{weinhold2020fast} shoots a ray from an interior point within the feasible region and marks the constraints first hit as non-redundant. To avoid solving optimization models, \cite{9314247} develops an efficient matrix operation, while \cite{hua2013eliminating} maps the variation of net demand to the binding situation of line limits. Yet both are heuristics without equivalence guarantee of original UC problems. Data-driven methods featuring the expressive representation are utilized in \cite{zhang2019data, pineda2020data,ardakani2018prediction, xavier2021learning, Cordero2022Warm-starting,sterzinger2023learning} to directly output redundant constraints through a trained model, which may be unreliable due to scale shifts or distribution shifts in UC instances. While LP-based screening provides an exact equivalence guarantee and strong generalization ability. 
% Naturally, concerns exist about the shortage of historical data and the infeasibility of the resulting UC model.
% the line flow along with 
% \vspace{-1em}
\subsection{Contribution}
Herein, we provide a novel vertex-guided perspective on LP-based screening to address the proposed research question. Our proposed anchor is the tie between the redundant constraints and the vertices of the screened region. Specifically, as shown in Fig. \ref{Fig-intro-constraint}, the vertices represent the intersections of the bounds defining the screened region. A redundant constraint is satisfied by all vertices and does not cut this region~\cite{balinski1961algorithm, mattheiss1973algorithm}, which is firstly applied to model predictive control problem in \cite{herceg2013multi}. Note that finding the vertices of the relaxed UC feasible region, which is a polytope, is nontrivial. In our proposed method, an LP is solved first, replacing the objective function of line flow value in \eqref{screening1} with UC's decision variables. This LP helps get each UC variable's upper and lower bounds over the screened region. Then, using such obtained bounds, an outer approximation—a hyperrectangle—containing the UC feasible region is built, as illustrated by the blue region in Fig. \ref{Fig-intro-constraint}. The intersections of these bounds define the vertices and their coordinates of interest. 

Once the vertices are ready, checking each constraint for every vertex is still cumbersome. To address such concern, through a scaling operation, we derive that the left-hand value of the screened constraint substituting the coordinate of any vertex is smaller than that substituting the variables' bounds. If the latter satisfies the constraint, this constraint will be redundant. Then, we can verify the target constraint only using the variables' bounds, avoiding vertex-by-vertex checks. Considering this procedure for each target constraint, a matrix can be formed, and all constraints can be screened simultaneously via a simple matrix operation. The whole screening process requires solving the LPs only for UC variables, e.g., generation dispatch. Since the number of generation units is typically far fewer than lines in the UC models, we can reduce the number of LPs to solve significantly compared to LFGS. 
% \cite{zhai2010fast, roald2019implied, porras2021cost, he2022enabling, he2023fast, telgen1983identifying, zhai2014fast, ding2020fast}.
% Finally, we implement a one-to-multiple screening method, where solving a single optimization model can identify multiple redundant constraints. 

A subsequent hurdle is that the outer approximation and the vertex-to-bound scaling may yield a loose screened region, potentially leading to insufficient constraint elimination. To address this, we consider an ensemble strategy: first applying faster vertex-guided screening (VGS) to remove most of the redundant constraints, then refining the results using tighter LFGS for the remaining constraints. Additionally, the improvements in \cite{roald2019implied, porras2021cost, he2022enabling, he2023fast} developed for LP-based screening, can be seamlessly integrated into our VGS. 
The main contributions of our work can be summarized as follows:
\begin{itemize}
\item [1)]
  A vertex-guided perspective is proposed for LP-based screening, whose computation cost depends on the number of decision variables rather than constraints. This leads to faster LP-based screening for UC problems.
  
  \item [2)]
  An ensemble of VGS and LFGS is applied to ensure sufficient constraint elimination. Further, the load operating range is involved to provide applicability across varying load inputs, while the prediction of UC cost and on/off commitment are integrated to remove more constraints. 
  % and can deal with more operation conditions than \cite{wang2023carbon,abdennadher2022carbon}.
  \item [3)]
   Neural network and K-Nearest-Neighbor (KNN) models are also investigated to improve our method and act as discrete variable predictors. Case studies on systems up to 2,383 buses validate the proposed methods' effectiveness with up to 8.8x acceleration compared to LFGS.
\end{itemize}
\section{Constraint Screening for UC Model} \label{Section_2}
This section details the formulation and assumptions of the investigated UC problem and illustrates the theoretical foundation for LP-based constraint screening.
\subsection{Unit Commitment Model}
In this paper, the system operators are supposed to decide both the ON/OFF status $\mathbf{u}$ and the dispatch level $\mathbf{x}$ for all generators to find the least operational costs corresponding to generators' cost vector $\mathbf{c}$. The power flows are modeled as a DC approximation, where $\mathbf{P}$ denotes the PTDF matrix \cite{hua2013eliminating}. $\overline{\mathbf{f}}$ and $\underline{\mathbf{f}}$ denote the upper and lower bounds of line flows. In practice, the UC problem may involve logic and ramp constraints, complicating the analysis of multi-period constraints. Since this work focuses on evaluating the acceleration benefit of our vertex-guided approach in LP-based screening, we follow \cite{pineda2020data, roald2019implied} and conduct our analysis on the single-period UC problem formulation:
\begin{subequations}
% \vspace{0.5em}
\label{UC}
\begin{align}
\min _{\mathbf{u}, \mathbf{x}}\quad & \mathbf{c}^{\top}\mathbf{x}\label{UC:obj_0}\\
\text { s.t. } \quad 
 &\underline{\mathbf{f}} \leq \mathbf{P}(\mathbf{B}\mathbf{x}-\boldsymbol{\ell}) \leq \overline{\mathbf{f}}, \label{UC:flow_0}\\
 &\sum_{g=1}^{|\mathcal{G}|}x_{g} - \sum_{n=1}^{|\mathcal{N}|}\ell_n =0, \quad \label{UC:balance_0}\\
 &u_g \underline{x}_g \leq x_{g} \leq u_g \bar{x}_g, \quad g=1,2,...,|\mathcal{G}|, \label{UC:gen_0}\\
 & u_g \in \{0, 1\}, \quad g=1,2,...,|\mathcal{G}|. \label{UC:u_0}
% // & x \in \boldsymbol{X}_T. \label{UC: temproal}
\end{align}
\end{subequations}
where $\mathcal{N}$ and $\mathcal{G}$ denote the index set of buses and generators respectively. $\mathbf{B} \in \mathbb{R}^{|\mathcal{N}| \times |\mathcal{G}| }$ maps the generation to each bus. \eqref{UC:flow_0}, \eqref{UC:balance_0} and\eqref{UC:gen_0} denote the line limit, the system power balance, and the generation bound, respectively. Note that in \eqref{UC:flow_0}, $\geq$ and $\leq$ are element-wise. \eqref{UC:u_0} enforces the binary constraint of the generator statuses, where $u_g=1$ indicates that the generator $g$ is on. $\ell_n$ denotes the load input at bus $n$. Here we denote the feasible region of \eqref{UC} as $R_{uc}$.

% Note the proposed method can be flexibly extended to handle the multi-period UC problem.

% $\ell_n$ denotes the net demand considering the load and the renewable generation. 
% \eqref{UC: temproal} denotes the temproal constraint set of generations such as the ramping constraints.
\subsection{Constraint Screening}
In real-world scenarios, the UC model can be large-scale MIP problems, which poses challenges in achieving a timely solution. Meanwhile, as shown in Fig. \ref{Fig-intro-constraint}, there exists a set of constraints in \eqref{UC}, which have no impact on UC solutions either adding/deleting them. This gives the potential to speed up the solution by eliminating such limits. 
% and let $j \in \mathcal{J}$
 % We can define the following definitions:
% \begin{Definition}
% (Inactive constraint) A constraint $j$ is an inactive constraint if and only if removing it from (or adding it to) $\mathcal{J}$ changes the feasible region while keeping the optimal solution $\boldsymbol{u}^*$ and $\boldsymbol{x}^*$ unchanged.
% \end{Definition}
\begin{Definition}
(Redundant Constraint) Let $\mathcal{J}$ be the index set corresponding to the constraints to be screened. A constraint $j$ is a redundant constraint if and only if the feasible region $R$ remains invariant when constraint $j$ is added or removed.
\end{Definition}
In the classic LP-based screening—LFGS, a model will be used to identify the redundant constraint, that is,
\begin{Definition}
(Screening Model) A screening model is formulated to maximize bi-direction line flow $f_{j}=\sum_{i=1}^{n}a_{i,j}(x_i-\boldsymbol{\ell}_i)$ over a feasible region that contains $R_{uc}$ with $\mathbf{u} \in [0,1]$ to get the actual bound $ f^{*,scr}_j$ for the target constraint $j$.
\end{Definition}
% \textcolor{blue}{you may need to clarify why the aim of screening model is to maximize line flow.} \textcolor{blue}{please briefly discuss the necessary and sufficient condition in [9]}. 
The actual bound $ f^{*,scr}_j$ in \cite{zhai2010fast} is obtained with the screening model \eqref{screening1} specialized to include \eqref{UC:flow_0}-\eqref{UC:gen_0}, exclude the target limit. The target limit is redundant to $R_{uc}$ when $ f^{*,scr}_j$ in their screening model does not exceed the predefined bound. We extend this criterion to general screening models and derive necessary corollaries for our analysis.
% , resulting in a feasible region $R_{LF}$
% The screening model in \cite{zhai2010fast} is specified as \eqref{screening2} along with Lemma \ref{lemma_1} proved,
% \begin{subequations} \label{screening2}
% \begin{align}
% \max_{\mathbf{u}, \mathbf{x}}~or \;\min _{\mathbf{u}, \mathbf{x}}\quad & f_j \label{scre_obj}\\
% \text { s.t. } \quad &\eqref{UC:gen_0}, \eqref{UC:balance_0}, \label{Screening2:ge&balance}\\
%  &\overline{\mathbf{f}}_{k} \leq \sum_{i=1}^{N}a_{i,k}(x_i-\boldsymbol{\ell}_i) \leq \overline{\mathbf{f}}_{k}, \quad k \neq j \label{Screening2:flow_0}\\
% & 0\leq u_i \leq 1, \quad i=1,2,...,N. \label{Screening2:u}
% \end{align}
% \end{subequations}
% where \eqref{Screening2:u} relaxes $u_i$ as the continuous variables and thus \eqref{screening2} can be solved as a linear programming problem with a larger feasible region (screened region) containing the original feasible region of UC model.

% \begin{lemma} \label{lemma_1}
% Let $ f^{*}_j$ denote the optimal objective value obtained in \eqref{screening2}. The upper bound or lower bound $j$ in \eqref{UC:flow_0} is proved to be a redundant constraint and can be removed from the original UC model when the following relationship holds,
% \begin{subequations}
% \begin{align}
%  f^*_j \leq \overline{f}_j \quad or \quad  f^*_j \geq \underline{f}_j
% \end{align}
% \end{subequations}   
% \end{lemma}

\begin{corollary} \label{c_1}
Let $M_{scr,relax}$ denote a screening model with the feasible region $R_{relax}$, and denote its constraint set as $\mathcal{J}_{relax}$. The upper bound (or lower bound) $j$ in \eqref{UC:flow_0} is redundant for $R_{relax}$  when the following relationship holds in $M_{scr,relax}$, 
\begin{align} \label{c1_check}
f^{*,scr}_j < \overline{f}_{j} \quad (or \quad  f^{*,scr}_j > \underline{f}_{j}).
\end{align}
\end{corollary}
\begin{proof} \label{c_1:proof}
For the case  $M_{scr,relax}$ is with constraint $j$, denote the region of $R_{relax}$ excluding constraint $j$ as $R_{relax}^{-j}$, and we have $R_{relax} \subseteq R_{relax}^{-j}$. Assume that constraint $j$ satisfying \eqref{c1_check} is non-redundant for $R_{relax}$, which means $R_{relax}^{-j} \neq R_{relax}$, i.e., $R_{relax}^{-j} \nsubseteq R_{relax}$. Then, the hyperplane $\mathcal{P}_j$ corresponding to constraint $j$ should be a facet of the polyhedral $R_{relax}$. Since $M_{scr,relax}$ is an LP, its optimal solution will be a vertex of $R_{relax}$. Note that $\mathcal{P}_j$ also represents the objective function of $M_{scr,relax}$, so this vertex will lie on $\mathcal{P}_j$, which indicates that the optimal value $f^{*,scr}_j = \overline{f}_{j}$ (or $\underline{f}_{j}$). Thus, \eqref{c1_check} will be violated, then we prove that constraint $j$ is redundant for $R_{relax}$ by this contradiction. 

For the case $M_{scr,relax}$ is without constraint $j$, we denote the region of $R_{relax}$ involving constraint $j$ as $R_{relax}^{+j}$, and we have $R_{relax}^{+j} \subseteq R_{relax}$. When \eqref{c1_check} holds, it means each feasible solution in $R_{relax}$ will satisfy constraint $j$, that is, $R_{relax} \subseteq R_{relax}^{+j}$. Thus, $R_{relax} = R_{relax}^{+j}$ holds and constraint $j$ is redundant for $R_{relax}$. 
\end{proof}
%它的逆否命题是正确的, 当一个约束不是redundant的时候, f^{*,scr}可能大于或者等于上下界, 等于的情况是因为被screen的约束binding了.
% 对于包含被screen的约束的screening model 必须是连续的, 不然推不出来小于上界就是redundant.
%不包含的, 不松弛也可以.
%这里没有充要条件是因为screening model的可行域没有排除被screen的约束.
% Consider a constraint $j$ that satisfies \eqref{c1_check} and assume it is non-redundant for $R_{uc}$. Then, there exists a solution of $M_{ori,scr}$ that makes $f^{*,scr}_j$ exceed $\overline{\mathbf{f}}_{j}$ or $\underline{\mathbf{f}}_{j}$. However, such a solution must be the feasible solution of $M_{scr, relax}$, and make \eqref{c1_check} violated. Thus, it can be proved that constraint $j$ is redundant by this contradiction.
% there will be solution satisfying the constraint set $\mathcal{J}_{relax}/{j}$ while let $f_j \geq \overline{\mathbf{f}}_j$ (or $f_j \leq \underline{\mathbf{f}}_j$). As the solution space is continuous, $f^{*,scr}_j 
% = \overline{\mathbf{f}}_{j} \quad (or \quad  f^{*,scr}_j 
% = \underline{\mathbf{f}}_{j})$ will occur for $M_{scr,relax}$
\begin{corollary} \label{C: transition}
Consider two regions $R_{scr,a}$ and $R_{scr,b}$, which satisfies $R_{scr,a} \subseteq R_{scr,b}$. When a constraint is identified as redundant for $R_{scr,b}$ via a screening model, it will be also redundant for $R_{scr, a}$.
\end{corollary}
\begin{proof}
Adding or removing constraint $j$ keeps $R_{scr,b}$ invariant, given that $R_{scr,a} \subseteq R_{scr,b}$, it follows that $R_{scr,a}$ also remains invariant. Thus, constraint $j$ is redundant for $R_{scr,a}$.
\end{proof}

LFGS can be cumbersome to optimize each line flow. Thus, we propose a vertex-guided way that only requires solving the LPs for each decision variable of the UC model.

\section{Vertex-Guided Constraint Identification} \label{Section_3}
This section discusses the relationship between the screened region's vertices and redundant constraints, and how to derive easily obtainable vertices through an outer approximation of the screened region. This relationship and the vertices are then used to guide the identification of redundant UC constraints with variables' bounds through a matrix operation.

% Due to the presence of binary variables $\boldsymbol{u}$, the feasible region $R_{uc}$ of the UC model may consist of multiple disconnected regions, as illustrated in Fig. \ref{Illustr: FR_3D_notation}.
% screening methods indicate that the constraints redundant for a relaxed region $R_{relax}$, which encompasses $R_{uc}$, will also be redundant for the original UC model.
% ignore the equality in \eqref{UC}, then
\subsection{Relationship Between Vertices and Redundant Constraints}
We first formulate a screened region to contain $R_{uc}$. By relaxing $\mathbf{u}$ to $\Tilde{\mathbf{u}} \in [0,1]$, the UC model can be converted to an LP, and the compact form can be given as follows,
\begin{subequations}\label{UC: u_relax}
\begin{align} 
\min_{\mathbf{y}:=[\mathbf{x,\Tilde{u}}]} \quad & \Tilde{\mathbf{c}}^T\mathbf{y}\\
s.t.~~~ &\Tilde{\mathbf{A}}\mathbf{y} \leq \Tilde{\mathbf{b}} \label{A_y}
         % &\mathbf{A}_{\Tilde{u}}\mathbf{\Tilde{u}} \leq \mathbf{b}_{\Tilde{u}} \label{A_u}
\end{align}    
\end{subequations}
where $\Tilde{\mathbf{A}}$, $\Tilde{\mathbf{b}}$ and $\Tilde{\mathbf{c}}$ are the coefficient matrix and vectors converted from \eqref{UC}. And denote feasible region of \eqref{UC: u_relax}  as $R_{\Tilde{\mathbf{u}}}$.
\begin{proposition}
The relaxed region $R_{\Tilde{\mathbf{u}}}$ will contain the original feasible region $R_{uc}$, that is, $R_{uc} \subseteq R_{\Tilde{\mathbf{u}}}$. According to Corollary \ref{C: transition}, a constraint redundant for $R_{\Tilde{\mathbf{u}}}$ is redundant for $R_{uc}$.
\end{proposition}

In the case study shown in Fig. \ref{Fig-intro-constraint}(b), the red hyperplane associated with a redundant constraint does not fit within any vertex of the screened region, indicating that it will be satisfied by all vertices. Theorem 1 formally elaborates on this scenario.
% Intuitively, as the red hyperplane shown in Fig. \ref{Fig-intro-constraint}, the hyperplane corresponding to a redundant constraint will not lie below any vertex of the screened region, meaning it will be satisfied by all vertices and this principle can be formalized mathematically as follows,%too abstract, specific constraints, example, unified color.
% \begin{lemma} \label{L: relax_u}
% Given that $R_{uc} \subseteq R_{\Tilde{u}}$, according to Corollary \ref{C: transition}, the redundant constraints for $R_{\Tilde{u}}$ will be redundant for $R_{uc}$.
% \end{lemma}
% \begin{proof}
% , Lemma \ref{L: relax_u} will hold.
% \end{proof}
\begin{theorem} \label{vetice_theorem}
Denote a vertex of the screened region $\hat{R}$ defined by matrix $\hat{\mathbf{A}}$ and vector $\hat{\mathbf{b}}$ as $\boldsymbol{v}_q$, and the vertex set as $\mathcal{Q}$. A constraint $\hat{\mathbf{A}}_j\mathbf{y} \leq \hat{b}_j $ is redundant for $\hat{R}$ when satisfying the following relationship,
\begin{align} \label{vertices_check}
\hat{\mathbf{A}}_j\boldsymbol{v}_q < \hat{b}_j, \; \forall q \in \mathcal{Q}.
\end{align}
\end{theorem}
\begin{proof}
We prove this by contradiction. For a constraint $\hat{\mathbf{A}}_j\mathbf{y} \leq \hat{b}_j $ satisfying \eqref{vertices_check}, we assume it can be non-redundant. According to Colloary \ref{c_1}, there is a solution $\mathbf{y}^*$ belonging to $\hat{R}$ which enables $\hat{\mathbf{A}}_j\mathbf{y}^* = b^*_j \geq \hat{b}_j$. The hyperplane $\hat{\mathbf{A}}_j\mathbf{y} = b^*_j$ through $\mathbf{y}^*$ will divide $\hat{R}$ into two parts, and the vertex $v_q$ lies in the half-space where $\hat{\mathbf{A}}_jv_q \geq b^*_j$ will satisfy $\hat{\mathbf{A}}_jv_q \geq \hat{b}_j$. This is contradicting the assumption $\hat{\mathbf{A}}_j\mathbf{y} \leq \hat{b}_j $ satisfying \eqref{vertices_check}. Clearly, $\hat{\mathbf{A}}_j\mathbf{y} \leq \hat{b}_j $ is redundant.
\end{proof} 
% \textbf{2) Necessity.} For a redundant constraint $\hat{\mathbf{A}}_j\mathbf{y} \leq \hat{\mathbf{b}}_j$, according to Corollary \ref{c_1}, for all solutions $\mathbf{y} \in \hat{Y}$ in $\hat{R}$ will satisfy $\Tilde{\mathbf{A}}_j\mathbf{y} \leq \Tilde{\mathbf{b}}_j$. Clearly, for each vertex $\boldsymbol{v}_q$, it will belong to $\hat{Y}$ and thus \eqref{vertices_check} holds.

Ideally, once we get all vertices of $R_{\Tilde{\mathbf{u}}}$, we can use Theorem \ref{vetice_theorem} to identify redundant constraints in the UC model. However, finding all vertices of the polytope $R_{\Tilde{\mathbf{u}}}$ is challenging \cite{khachiyan2009generating}. Moreover, it requires identifying non-redundant constraints first, which contradicts the process of Theorem \ref{vetice_theorem} that needs to obtain the vertices first. Thus, we consider finding an outer approximation for $R_{\Tilde{\mathbf{u}}}$ whose vertex can be found without knowing the binding situations of the constraints.

% \textcolor{blue}{due to what}. 

% \textcolor{blue}{discuss more about Theorem 1. For example, what is the difficulty to maintain the ideal situation? Is it possible to use Theorem 1 in practice?}
% why?
\subsection{Outer Approximation based on Actual Variables' Bounds}
\begin{figure}[tb] 
\vspace{-1.5em}
    % \hspace{0.35cm}
	\centering
	\includegraphics[scale=0.8]{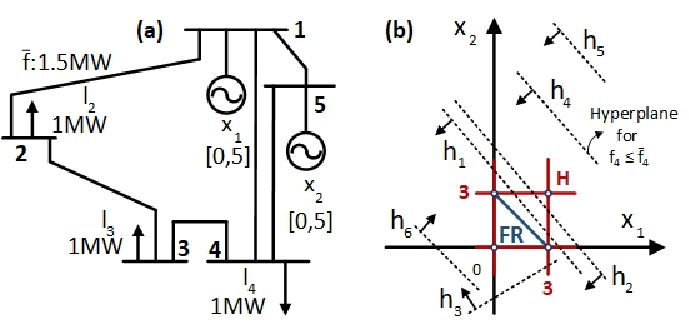}
    % \vspace{-0.5em}
	\caption{Toy example for VGS. \textbf{(a)} is a 5-bus system with 2 units and 6 lines. To visualize the outer approximation, we let $\mathbf{u}=\mathbf{1}$, and then this UC instance only has two variables $x_1$ and $x_2$ whose solution space is given in \textbf{(b)}. The red lines are the bounds of $x_1$ and $x_2$ obtained by \eqref{UC: y_bound}. They formulate the hyperrectangle \textcolor{Maroon}{\textbf{H}} containing the UC feasible region \textcolor{NavyBlue}{\textbf{FR}}. It shows the limits of $f_1-f_2$ are under the vertex (3,3), which means \eqref{vertices_check} is violated, while for line limits $f_3-f_6$, \eqref{vertices_check} holds and thus they are redundant for \textcolor{Maroon}{\textbf{H}}.} \label{GE:5bus}
\vspace{-1.5em}    
\end{figure}
% x1 intercept = 3.4815    4.4512    5.0939    8.6383   13.7952   -3.0473
% x2 intercept = 4.2317    5.4104   -3.6129   10.4998   16.7679   -2.3372
Here we consider a hyperrectangle as the outer approximation for $R_{\Tilde{\mathbf{u}}}$, whose facets are the actual decision variables' bounds over $R_{\Tilde{\mathbf{u}}}$. Then, the vertices of the hyperrectangle can be directly given by the intersection of these bounds and replace the vertices of $R_{\Tilde{\mathbf{u}}}$ in Theorem \ref{vetice_theorem}. For each variable $y_p$, we can find its actual upper bound $\overline{y}_p$ or lower bound $\underline{y}_p$ by the following LP model,
\begin{subequations}\label{UC: y_bound}
\begin{align} 
\max_{\mathbf{y}} \quad or \quad  \min_{\mathbf{y}}& \quad y_p\\
s.t.~~~ &\Tilde{\mathbf{A}}\mathbf{y} \leq \Tilde{\mathbf{b}} \label{A_y_bound}
         % &\mathbf{A}_{\Tilde{u}}\mathbf{\Tilde{u}} \leq \mathbf{b}_{\Tilde{u}} \label{A_u}
\end{align} 
\end{subequations}
we solve \eqref{UC: y_bound} for all variables to get $\overline{\mathbf{y}}$ and $\underline{\mathbf{y}}$. Then we can formulate a hyperrectangle $R_{bound}$ defined as follows, 
\label{UC: y_bound_range}
\begin{align} 
\underline{\mathbf{y}} \leq \mathbf{y} \leq \overline{\mathbf{y}}.
\end{align}    
\begin{proposition} \label{R_bound_contain}
The hyperrectangle $R_{bound}$ will contain the relaxed region $R_{\Tilde{u}}$ and the original feasible region $R_{uc}$.
\end{proposition}
% \begin{proof}
% to do
% \end{proof}

% moce the following part to appendix
% \begin{proof}
% Assume that $R_{\Tilde{u}} \not\subset R_{bound}$, which means $\exists\Tilde{\mathbf{y}} \in R_{\Tilde{u}}$, satisfies $\Tilde{\mathbf{y}} \notin R_{bound}$. Thus, $\exists p \in P$, $\Tilde{y}_p > \overline{y}_p$, or $\Tilde{y}_p < \underline{y}_p$. However, $\forall \mathbf{y} \in R_{\Tilde{u}}$, $ 
%  \underline{y}_p \leq y_p \leq \overline{y}_p$ which is contradictory to the assumption. Thus, $R_{\Tilde{u}} \subseteq R_{bound}$ can be proved. Since $R_{uc} \subseteq R_{\Tilde{u}}$, $R_{uc} \subseteq R_{bound}$ can be obtained.
% \end{proof}

The following arguments also validate the feasibility of the proposed outer approximation method.

\begin{corollary} \label{R_bound_redundant}
Given that $R_{uc} \subseteq R_{bound}$ and according to Corollary \ref{C: transition}, the redundant constraints for $R_{bound}$ will be redundant for $R_{uc}$.
\end{corollary}

Then, we can use the vertices of the hyperrectangle $R_{bound}$ to screen out the redundant constraints for the UC model. In  Fig. 
\ref{GE:5bus}, we showcase a 5-bus system as a case study to illustrate this process. Specifically, we use \eqref{UC: y_bound} to get the generation bounds for two units and then use them to build a hyperrectangle $\mathbf{H}$ to identify the redundant upper bounds of line flows. Four upper bounds satisfy Theorem \ref{vetice_theorem} for $\mathbf{H}$ and can thus be safely screened out.

% \textcolor{blue}{you may need to briefly describe this example.}

\renewcommand{\algorithmicrequire}{\textbf{Input:}}
  \renewcommand{\algorithmicensure}{\textbf{Output:}}
  \begin{algorithm}[t] 
    \caption{\underline{E}nsemble \underline{O}f \underline{V}GS and \underline{L}FGS (EOVL)}
    \begin{algorithmic}[1] \label{Algorithm: OACS}
      \REQUIRE UC model \eqref{UC}, the index set $\mathcal{J}$ of line limits and decision variable index set $\mathcal{P}$, system parameters.
      \ENSURE Redundant constraint set $\mathcal{J}^{-}$, reduced UC model with constraint set $\mathcal{J}^{+}=\mathcal{J}-\mathcal{J}^{-}$.
      \renewcommand{\algorithmicensure}{\textbf{Initialize:}}
      \ENSURE  Let $\mathcal{J}^{-}=\varnothing$, $\mathcal{J}^{+}=\varnothing$, the bounds of the decision variables $\overline{\mathbf{y}} =: [\overline{\boldsymbol{x}}, \overline{\boldsymbol{u}}], \underline{\mathbf{y}} =: [\underline{\boldsymbol{x}}, \underline{\boldsymbol{u}}]$, and $j=1$.\\
     $\Box$ \textit{\textbf{Outer-Approximation:}}
      \\
     \STATE Let $\boldsymbol{u} \in \{0,1\} \Rightarrow \Tilde{\boldsymbol{u}}\in [0,1]$: \\ \quad Original UC model \eqref{UC} $\Rightarrow$ Relaxed UC model \eqref{UC: u_relax},
      \\ \quad Constraint set $\mathbf{A}\mathbf{y} \leq \mathbf{b} \Rightarrow \Tilde{\mathbf{A}}\mathbf{y} \leq \Tilde{\mathbf{b}}$.\\
    \STATE Solve model \eqref{UC: y_bound} with $\Tilde{\mathbf{A}}\mathbf{y} \leq \Tilde{\mathbf{b}}$: 
       \\ \quad Get maximum $\overline{y}^{*}_p$ and minimum $\underline{y}^{*}_p$, $\forall p \in \mathcal{P}$.
       \\
    \STATE  Update $\overline{\mathbf{y}}$ and $\underline{\mathbf{y}}$ with $\overline{y}^{*}_p$ and $\underline{y}^{*}_p$, $\forall p \in \mathcal{P}$.\\
    \STATE  Outer-approximation for the UC feasible region: \\~~~~~~~~~~~~~~~~~~~~$\underline{\mathbf{y}} \leq \mathbf{y} \leq \overline{\mathbf{y}}$
    \\
    $\Box$ \textit{\textbf{Screening-Vertex Guided}}\\
    \STATE Let $\mathbf{\omega} = \epsilon(\Tilde{\mathbf{A}})\circ\Tilde{\mathbf{A}}*(\overline{\mathbf{y}}-\underline{\mathbf{y}})+\Tilde{\mathbf{A}}*\underline{\mathbf{y}}-\Tilde{\mathbf{b}}$.\\
    \WHILE {$ j \leq |\mathcal{J}|$}
      \STATE If ${\omega}_{j} < 0$: $\mathcal{J}^{-}=\mathcal{J}^{-}\cup\{j\}$.\\
      \STATE $j = j+1$
      \ENDWHILE
    \\
    $\Box$ \textit{\textbf{Screening-Line Flow Guided} (Optional):}
      \WHILE {$\mathcal{J}^{-}+\mathcal{J}^{+}\neq \mathcal{J}$}
      \STATE $\Tilde{\mathcal{J}} = \mathcal{J}-\mathcal{J}^{-}-\mathcal{J}^{+}$.\\
      \STATE Maximize ($\Tilde{\mathbf{A}}_j\mathbf{y}$) s.t. ($\Tilde{\mathbf{A}}_{\mathcal{J}/j}\mathbf{y}\leq\Tilde{b}_{\mathcal{J}/j}$), $j \in \Tilde{\mathcal{J}}$.
      \STATE If $(\Tilde{\mathbf{A}}_j\mathbf{y})^* > \Tilde{b}_j$: $\mathcal{J}^{+}=\mathcal{J}^{+}\cup\{j\}$, \\else: $\mathcal{J}^{-}=\mathcal{J}^{-}\cup\{j\}$.
      \ENDWHILE
      \STATE Return $\mathcal{J}^{-}$.
    \end{algorithmic}
  \end{algorithm}

\subsection{Fast Identification by Bridging Vertices and Bounds}
The vertex $\boldsymbol{v}_q$ of $R_{bound}$ can be determined by the bounds $\overline{\mathbf{y}}$ and $\underline{\mathbf{y}}$ as follows,
\begin{align} \label{vertices}
  \boldsymbol{v}_q=:(v_{q,1},v_{q,2},..,v_{q,2|\mathcal{G}|}), v_{q,p} \in \{\underline{y}_p, \overline{y}_p\}
\end{align}

However, there are numerous vertices, and it is cumbersome to check the condition \eqref{vertices_check} for each vertex and each constraint. Thus, we present a fast criterion for identifying all redundant constraints through a matrix operation, which is derived by bridging the vertices and bounds via \eqref{vertices_check}. We use $\Phi(\cdot)$  to represent applying unit-step function $\phi(\cdot)$ to each element in the matrix, and $\circ$ denotes element-wise multiplication.

\begin{theorem} \label{bound__operation}
Let $\boldsymbol{{\omega}} = \Phi(\Tilde{\mathbf{A}})\circ\Tilde{\mathbf{A}}(\overline{\mathbf{y}}-\underline{\mathbf{y}})+\Tilde{\mathbf{A}}\underline{\mathbf{y}}-\Tilde{\mathbf{b}}$. If ${\omega}_j < 0$, then $\Tilde{\mathbf{A}}_j\mathbf{y} \leq \Tilde{b}_j$ associated with the $j$-th constraint is redundant for $R_{uc}$.
\end{theorem}

% \textcolor{blue}{move this text into the theorem}
\begin{proof}
Consider a constraint $\Tilde{\mathbf{A}}_j\mathbf{y} \leq \Tilde{b}_j$, \eqref{vertices_check} can be reformulated as,
\begin{align} \label{A_j}
 \sum_{p=1}^{2|\mathcal{G}|} a_{j,p}v_{q,p} < \Tilde{b}_j, \forall q \in \mathcal{Q};
\end{align}
where $a_{j,p}$ is the coefficient for the $p$-th decision variable in the $j$-th constraint. 

For the coefficient $a_{j,p} \leq 0$, the following inequality holds,
\begin{align} \label{negative}
a_{j,p}\overline{y}_p \leq a_{j,p}v_{q,p} \leq a_{j,p}\underline{y}_p,
\end{align}
while for the coefficient $a_{j,p} \geq 0$, we have, 
\begin{align} \label{positive_1}
a_{j,p}\underline{y}_p \leq a_{j,p}v_{q,p} \leq a_{j,p}\overline{y}_p,
\end{align}
which can be converted to,
\begin{align} \label{positive_2}
a_{j,p}\underline{y}_p \leq  a_{j,p}v_{q,p} \leq a_{j,p}(\overline{y}_p-\underline{y}_p) + a_{j,p}\underline{y}_p,
\end{align}
combining \eqref{negative} and \eqref{positive_2}, the following relationship can hold for any vertice $\boldsymbol{v}_q$,
\begin{align} \label{positive_all}
\sum_{p=1}^{2|\mathcal{G}|} a_{j,p}v_{q,p} - \Tilde{b}_j \leq \sum_{p=1}^{2|\mathcal{G}|} \phi(a_{j,p})(\overline{y}_p-\underline{y}_p)+a_{j,p}\underline{y}_p - \Tilde{b}_j.
\end{align}

The right-hand side can be further derived for the whole constraint set as follows,
\begin{align} \label{reundant_check}
\boldsymbol{\omega} = \Phi(\Tilde{\mathbf{A}})\circ\Tilde{\mathbf{A}}(\overline{\mathbf{y}}-\underline{\mathbf{y}})+\Tilde{\mathbf{A}}\underline{\mathbf{y}}-\Tilde{\mathbf{b}},
\end{align}
if $\omega_j < 0$, the following inequation can hold for all vertices,
\begin{align} \label{screened_con}
\sum_{p=1}^{2|\mathcal{G}|} a_{j,p}v_{q,p} < \Tilde{b}_j;
\end{align}
according to Theorem \ref{vetice_theorem} and Corollary \ref{R_bound_redundant}, the constraint $\Tilde{\mathbf{A}}_j\mathbf{y} \leq \Tilde{b}_j$ will be redundant for $R_{bound}$ and $R_{uc}$. 
\end{proof}

\subsection{Computational Complexity Analysis}
The computational complexity of VGS arises from the LP model and the matrix operation. Let $O(\mathcal{C}_{LP})$ denote the complexity of solving a LP, where $\mathcal{C}_{LP} \in [|\mathcal{G}|^3, 4^{|\mathcal{G}|}]$ depending on the solver. Since we need to solve $2|\mathcal{G}|$ LPs, the total complexity for this process is $O(|\mathcal{G}|\mathcal{C}_{LP})$. Regarding the matrix operation in \eqref{reundant_check}, its complexity is in the order of $O(|\mathcal{G}||\mathcal{J}|+|\mathcal{G}|)$. Then the overall complexity of VGS is in the order of $(|\mathcal{G}|\mathcal{C}_{LP}+|\mathcal{G}||\mathcal{J}|+|\mathcal{G}|)$, which can be denoted as $O(|\mathcal{G}|\mathcal{C}_{LP})$. 

Here we further compare the computational cost of VGS and LFGS. The LP model in VGS has $2|\mathcal{G}|$ variables and $4|\mathcal{G}|+|\mathcal{J}|+2$ constraints, while LFGS has $2|G|$ variables and  $4|\mathcal{G}|+|\mathcal{J}|+1$ constraints. Thus, the scale of LP for both methods is comparable, and the computational complexities of their LPs are equivalent, i.e., $O(\mathcal{C}_{LP})$. For LFGS, we need to solve $|J|$ LPs, so its computation complexity is $O(|J|\mathcal{C}_{LP})$. In the UC problems, $|G|$ is generally much smaller than $|\mathcal{J}|$. This means conducting VGS is faster than LFGS, while LFGS may have a tighter screened region and remove more constraints. 

To match LFGS performance while accelerating the screening, we use an ensemble way that applies VGS first, followed by LFGS on the remaining line limits, and we term this \underline{E}nsemble \underline{O}f \underline{V}GS and \underline{L}FGS method as EOVL described in Algorithm \ref{Algorithm: OACS}. Denote the number of removed limits by VGS as $N_{V}$. Then, the complexity of EOVL is $O((|G|+(|\mathcal{J}|-N_{V}))\mathcal{C}_{LP})$, while the worst case is $O((|G|+|J|)\mathcal{C}_{LP})$. Furthermore, we can estimate the execution time ratio between LFGS and EOVL by $\frac{2|G|\Delta t_{V}+(|\mathcal{J}|-N_{V})\Delta t_{LF}}{|\mathcal{J}|\Delta t_{LF}}$. Thus, as $\Delta t_{V}\approx \Delta t_{LF}$ and $|N_{V}| > 2|G|$, we can still  realize a considerable screening acceleration. 

% \textcolor{red}{Refer to Algorithm 1 and denote the proposed method as EOVL?}

% Here we further compare the computational cost of VGS and LFGS. The LP model in VGS has $2|\mathcal{G}|$ variables and $4|\mathcal{G}|+|\mathcal{J}|+2$ constraints, while LFGS has $2|G|$ variables and  $4|\mathcal{G}|+|\mathcal{J}|+1$ constraints. Thus, the LP scale for both methods is similar, and the computational complexities of their LPs are equivalent, i.e., $O(\mathcal{C}_{LP})$. For LFGS, we need to solve $|J|$ LPs, so its computation complexity is $O(|J|\mathcal{C}_{LP})$. In the UC problems, $2|G|$ is generally much smaller than $|\mathcal{J}|$. This means VGS can be faster than LFGS, while LFGS may have a tighter screened region and remove more constraints. To match LFGS performance while accelerating the screening, we propose an ensemble screening (ES) that applies VGS first, followed by LFGS on the remaining line limits. Denote the number of removed limits by VGS as $N_{V}$. Then, the complexity of the ensemble way is $O((2|G|+(|\mathcal{J}|-N_{V}))\mathcal{C}_{LP})$, while the worst case is $O((2|G|+|J|)\mathcal{C}_{LP})$. Further, the screening time ratio of LFGS and ES can be approximated by $\frac{2|G|\Delta t_{V}+(|\mathcal{J}|-N_{V})\Delta t_{LF}}{|\mathcal{J}|\Delta t_{LF}}$. Thus, when $\Delta t_{V}\approx \Delta t_{LF}$ and $|N_{V}| > 2|G|$, we can realize an screening acceleration.

\section{Improvement for Vertex-Guided Screening} \label{Section_4}
Different variables' bounds create heterogeneous outer approximations and vertices, leading to varying redundancy statuses $\boldsymbol{\omega}$ of constraints in \eqref{reundant_check}. This section explores ways to further improve the vertex-guided screening regarding load variation applicability and the sufficiency of redundancy removal by strategically adjusting the variables' bounds.

% This section explores ways to obtain valid screening results across varying load profiles or remove additional redundancy by strategically adjusting the outer approximation. 
\begin{figure*}[t] 
    \vspace{-1em}
    % \hspace{1em}
	\centering
\includegraphics[scale=0.7]{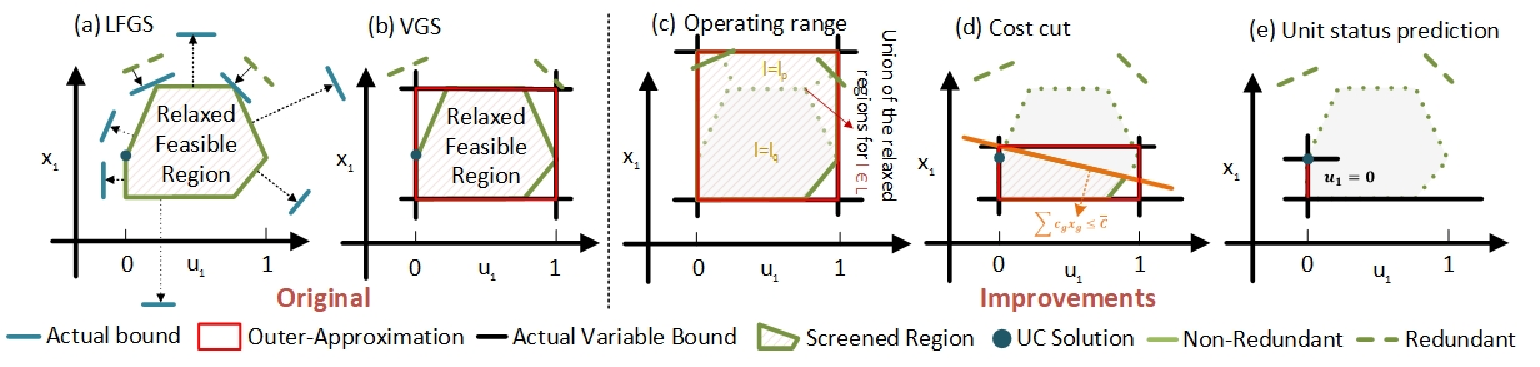}
 \vspace{-2em}
 \caption{Illustration of the LP-based screening schemes. (a) and (b) showcase that VGS can identify the same redundant constraints using fewer bounds given by solving LPs. In (c), the screened region expands to the union of screened regions for all $\boldsymbol{l} \in \mathcal{L}$, removing fewer constraints while bypassing screening for each $\boldsymbol{l} \in \mathcal{L}$. (d) and (e) showcase how the cost bound or unit status prediction further cuts the screened regions to remove more constraints.}
	\label{GE:fine-tuning}
 \vspace{-1em}
% \textcolor{red}{Change the color coding; add dashed lines; what does the current shaded region in second figure mean?}
% \textcolor{red}{In each Subsection of Sec. IV, denote which subfigure it is corresponding to.}
\end{figure*}
\subsection{Integration of Load Operating Range}
Empirical evidence indicates that UC instances associated with certain load operating ranges may share partial redundant constraints in their respective load input. Thus, by means of screening constraints at the load operating range $\mathcal{L}$ rather than a single load sample, we can find such intersection of redundant constraints, and bypass screening for each UC instance with specific load inputs \cite{ardakani2013identification, roald2019implied, awadalla2022influence}. To implement this, the load inputs $\boldsymbol{\ell}$ are included as the decision variable in the LP screening model as follows,
\begin{subequations}\label{UC: y_bound_range}
\begin{align} 
\max_{\mathbf{y}, \boldsymbol{\ell}} \quad or \quad   \min_{\mathbf{y}, \boldsymbol{\ell}}  \quad & y_p\\
s.t.\;\; &\Tilde{\mathbf{A}}\mathbf{y} \leq \Tilde{\mathbf{b}} \label{A_y_bound}
\\ & \boldsymbol{\ell} \in \mathcal{L}.
\end{align}    
\end{subequations}
Model \eqref{UC: y_bound_range} aims to find the upper bound and lower bound for each UC decision variable $y_p$ over  $\mathcal{L}$. With these bounds, a hyperrectangle $R_{bound, \mathcal{L}}$ is built as shown in Fig. \ref{GE:fine-tuning}(c). $R_{bound, \mathcal{L}}$ is supposed to identify the mutual redundant constraints for $\boldsymbol{\ell} \in \mathcal{L}$, as illustrated in Proposition 3.
% $R_{bound, \mathcal{L}}$ contains all the feasible solutions $\hat{\mathbf{y}}$ regarding $\mathcal{L}$.
\begin{proposition} \label{operating_range}
The redundant constraints for $R_{bound, \mathcal{L}}$ will be redundant for the feasible region $R_{uc, \hat{\boldsymbol{l}}}$ of the UC instance with $\hat{\boldsymbol{l}} \in \mathcal{L}$. 
% "feasible" is not necessary for load.
\end{proposition}
\begin{proof}
Denote the optimal objective of \eqref{UC: y_bound_range} for each decision variable $y_p$ as $\overline{y}^*_p(\overline{\boldsymbol{l}}_p^*)$ and $\underline{y}^*_p(\underline{\boldsymbol{l}}_p^*)$, where $\overline{\boldsymbol{l}}_p^*$ and $\underline{\boldsymbol{l}}_p^*$ are the optimal solution for the load variables, then for any feasible $\boldsymbol{l} \in \mathcal{L}$, we have $\overline{y}^*_p(\boldsymbol{l})\leq \overline{y}^*_p(\overline{\boldsymbol{l}}_p^*)$ and $\underline{y}^*_p(\boldsymbol{l})\geq \underline{y}^*_p(\underline{\boldsymbol{l}}_p^*)$. For the specific load inputs $\hat{\boldsymbol{l}} \in \mathcal{L}$, we denote the hyperrectangle built by $\overline{y}^*_p(\hat{\boldsymbol{l}})$ and $\underline{y}^*_p(\hat{\boldsymbol{l}})$ from \eqref{UC: y_bound} as $R_{bound,\hat{\boldsymbol{l}}}$. For each point $\boldsymbol{y} \in R_{bound,\hat{\boldsymbol{l}}}$, it will satisfy $y_p \leq y^*_p(\hat{\boldsymbol{l}}) \leq y^*_p(\boldsymbol{l}_p^*)$. Thus, we have $R_{bound,\hat{\boldsymbol{l}}} \subseteq R_{bound,\mathcal{L}}$. 

Based on Proposition \ref{R_bound_contain} and Corollary \ref{R_bound_redundant}, it can be further proved that $R_{uc, \hat{\boldsymbol{l}}} \subseteq R_{bound,\hat{\boldsymbol{l}}} \subseteq R_{bound,\mathcal{L}}$ and thus a redundant constraint for $R_{bound,\mathcal{L}}$ will be redundant for $R_{uc, \hat{\boldsymbol{l}}}$.
\end{proof}
%对于不可行的load input, 可行的load input对应的feasible regiongs构成的解空间内找不到它对应的解, 解空间外的点必然违反部分构成这些解空间的active constraints, 所以对于这些解空间来说redundant的constraints, 也无法让对应的load input对应的feasible reion从空集变成非空, 因此infeasible load input来说, 依然是redundant的.
Then, by applying Theorem \ref{bound__operation} on $R_{bound, L}$, we can find the redundant constraints for $R_{bound, L}$, which will be redundant for $\forall \boldsymbol{\ell} \in \mathcal{L}$. This shows that VGS for an operating range can provide applicability across varying load inputs.

\subsection{Integration of Cost Cut}
Fig. \ref{Fig-intro-constraint} indicates a smaller $R_{bound}$ derived from a tighter $R_{\Tilde{\mathbf{u}}}$ can find more redundant constraints, motivating us to tighten $R_{\Tilde{\mathbf{u}}}$ strategically. \cite{porras2021cost} observes that the solutions belong to $R_{\Tilde{\mathbf{u}}}$ or $R_{uc}$ yield varying UC costs. A cutting plane corresponding to a specific UC cost $\overline{C}$ can then be used to tighten $R_{uc}$ and $R_{\Tilde{\mathbf{u}}}$  as shown in Fig. \ref{GE:fine-tuning}(d), ensuring all remaining solutions have costs below $\overline{C}$. Integrating such cost cut, \eqref{UC: y_bound} can be developed as follows,
\begin{subequations}\label{UC: y_bound_cost}
\begin{align} 
\max_{y} \quad or \quad  & \min_{y}  \quad y_p\\
s.t.~ &\Tilde{\mathbf{A}}\mathbf{y} \leq \Tilde{\mathbf{b}} \\
&\sum_{g=1}^{|\mathcal{G}|} c_g x_g \leq \overline{C} \label{Screening: cost bound}
\end{align} 
\end{subequations}
% \begin{align} 
% \sum_{g=1}^{|\mathcal{G}|} c_g x_g \leq \overline{C} \label{Screening: cost bound}
% \end{align}  
To ensure UC feasibility and prevent the optimal UC solution from being excluded, it is necessary to ensure $\overline{C}$ is above the optimal UC cost $C^*$. 
\begin{proposition} \label{VGS:cost}
Given $\overline{C} \geq C^*$, the outer approximation $R_{bound}$ considering $\overline{C}$ can be tightened to $R_{bound, \overline{C}}$. The screening on $R_{bound, \overline{C}}$ will not change the optimal solution $\boldsymbol{y}^*_{uc}$ and possibly remove more constraints.
\end{proposition}

\begin{proof}
Denote the feasible region of model \eqref{UC: y_bound_cost} as $R_{\Tilde{u}, \overline{C}}$. By integrating \eqref{Screening: cost bound} into the UC model \eqref{UC}, we obtain a tighter UC feasible region $R_{uc, \overline{C}} \subseteq R_{uc}$. Given $\overline{C} \geq C^*$, we have $\boldsymbol{y}^*_{uc} \in R_{uc, \overline{C}}$. Thus, $\boldsymbol{y}^*_{uc}$ is optimal for both $R_{uc}$ and $R_{uc, \overline{C}}$. $R_{bound, \overline{C}}$ is the outer approximation of the region with relaxed $\mathbf{u}$, denoted as $R_{\Tilde{u}, \overline{C}}$, according to Proposition \ref{R_bound_contain}, we have $R_{uc, \overline{C}} \subseteq R_{\Tilde{u}, \overline{C}} \subseteq R_{bound, \overline{C}}$. According to Corollary \ref{R_bound_redundant}, the redundant constraints for $R_{bound, \overline{C}}$ will be redundant for $R_{uc, \overline{C}}$ and can be removed without changing $R_{uc, \overline{C}}$ so that remain $\boldsymbol{y}^*_{uc}$ optimal.

Since $R_{\Tilde{u}, \overline{C}} \subseteq R_{\Tilde{u}}$, we have $\overline{\mathbf{y}}_{\overline{C}} \leq \overline{\mathbf{y}}$ and $\underline{\mathbf{y}}_{\overline{C}} \geq \underline{\mathbf{y}}$, which indicates $\boldsymbol{\omega}_{\overline{C}} \leq \boldsymbol{\omega}$. Then, compared to the screening on $R_{bound}$, more constraints which satisfy $\omega_{j,\overline{C}} < 0, \omega_{j} \geq 0$ can be removed. 
\end{proof}

% since $R_{\Tilde{u}, \overline{C}} \subseteq R_{\Tilde{u}}$, we have $\overline{\mathbf{y}}_{\overline{C}} \leq \overline{\mathbf{y}}$ and $\underline{\mathbf{y}}_{\overline{C}} \geq \underline{\mathbf{y}}$, which indicates $R_{bound,\overline{C}} \subseteq R_{bound}$.
% while can be removed by the screening for $R_{bound,\overline{C}}$

Therefore, carefully selecting $\overline{C}$ is essential to eliminate as many constraints as possible while ensuring the optimal UC solution remains unchanged. Herein, for a specific load sample $\boldsymbol{l}$, we use a trained neural network $f_{NN}(\boldsymbol{\ell})$ to predict the UC costs. Finally, we introduce a relaxation parameter $\epsilon$ to account for prediction uncertainties, expressed as $f_{NN}(\boldsymbol{\ell})(1+\epsilon)$. This relaxed prediction is intended to slightly exceed the actual cost $C^*$ and serve as an upper limit $\overline{C}$.

% \textcolor{blue}{whats the meaning of "a relaxed cost $f_{NN}(\boldsymbol{\ell})(1+\epsilon)$ can be given as $\overline{C}$"}.

% For a load region $L$, we will follow \cite{10298823} to use a basic linear fitted model $a_0+b_0D$  to represent the mapping from the total load $D$ to the cost bound $\overline{C}$.
 
\subsection{Integration of Unit Status Prediction}
Besides the UC cost, we can also utilize the information on the historical unit status situation, and further restrict the bounds used in \eqref{UC: y_bound} to remove more constraints. Specifically, as depicted in Fig. \ref{GE:fine-tuning}(e), we can add a cutting plane \eqref{Screening: unit bound} corresponding to the predicted on/off status for partial units as follows,
\begin{align}  
u(k) = \hat{u}(k), \; k \in \mathcal{K} \label{Screening: unit bound}
\end{align}
where $\hat{u}(k)$ is the predicted on/off status of unit $k$. $\mathcal{K}$ denotes the index set of predicted units and $\mathcal{K} \subseteq \mathcal{G}$ holds. We will add \eqref{Screening: unit bound} to the UC model \eqref{UC} to get the final UC solution in the feasible region $R_{uc,\hat{u}(k)}$ and the outer-approximation model \eqref{UC: y_bound} to find the tightened variables' bounds over $R_{\Tilde{u}, \hat{u}(k)}$ for building the hyperrectangle $R_{bound, \hat{u}(k)}$.  
\begin{proposition} \label{VGS:unit}
The redundant constraints for $R_{bound, \hat{u}(k)}$ will be redundant for $R_{uc,\hat{u}(k)}$.
\end{proposition}

\begin{proof}
Given that  $R_{bound, \hat{u}(k)}$ is the outer approximation of $R_{\Tilde{u} , \hat{u}(k)}$ and $R_{uc, \hat{u}(k)} \subseteq R_{\Tilde{u}, \hat{u}(k)}$, according to Proposition \ref{R_bound_contain}, we have $R_{uc, \hat{u}(k)} \subseteq R_{\Tilde{u}, \hat{u}(k)} \subseteq R_{bound, \hat{u}(k)}$. Thus, according to Corollary \ref{R_bound_redundant}, the redundant constraints for $R_{bound, \hat{u}(k)}$ will be redundant for $R_{uc, \hat{u}(k)}$.
\end{proof}

Thus, by applying Theorem \ref{bound__operation} on $R_{bound, \hat{u}(k)}$, we can find the constraints redundant for $R_{uc,\hat{u}(k)}$ reliably. Note these screened may involve more constraints that are non-redundant for $R_{uc}$, while $\hat{u}(k)$ can be non-optimal. This suggests a tradeoff between solution accuracy and screening efficiency.

\begin{table}[t]
\vspace{-1.5em}
% \hspace{0.1cm}
\centering
\caption{Settings of Benchmarked Constraint Screening Schemes} \label{VGS: Schemes}
\setlength{\tabcolsep}{1.8mm}{
\begin{tabular}{cccccc}
\hline
Scheme & LFGS & VGS & Load Range                         & Cost Cut                         & Unit Status Cut                         \\ \hline
S1     &     &   $\surd$        &                             &                             &                             \\
S2     &     $\surd$    & &                             &                             &                             \\
S3     & $\surd$    & $\surd$   &                             &                             &                             \\
S4     & $\surd$    & $\surd$   & $\surd$                     &                             &                             \\
S5     & $\surd$    & $\surd$   &                             & $\surd$                     &                             \\
S6     & $\surd$    & $\surd$   &                             &                             & $\surd$                     \\ 
S7     & $\surd$    & $\surd$   &                          &     $\surd$                           & $\surd$                     \\     
\hline
\end{tabular}}
\end{table}
% Using this information, Fig. 4 presents a Venn diagram of UCD tightening approaches to reflect the various strengthened relaxations considered here and the types of screened constraints. %(总结一下待测试的模型)

% \vspace{-1.4em}
\section{Case Studies} \label{Section_5}
This section evaluates the screening speed and sufficiency of our VGS approach along with its enhancements, particularly in performance relative to LFGS. Especially, the NN models and KNN models are utilized to predict the UC cost cut and unit status respectively. We also open source our approach and details are given at \url{https://github.com/Hexuan085/UC_VGS}.

\subsection{Simulation Setups}
\begin{figure*}[b]  
\vspace{-1em}
    % \hspace{1em}
	\centering
\centering
\includegraphics[scale=0.45]{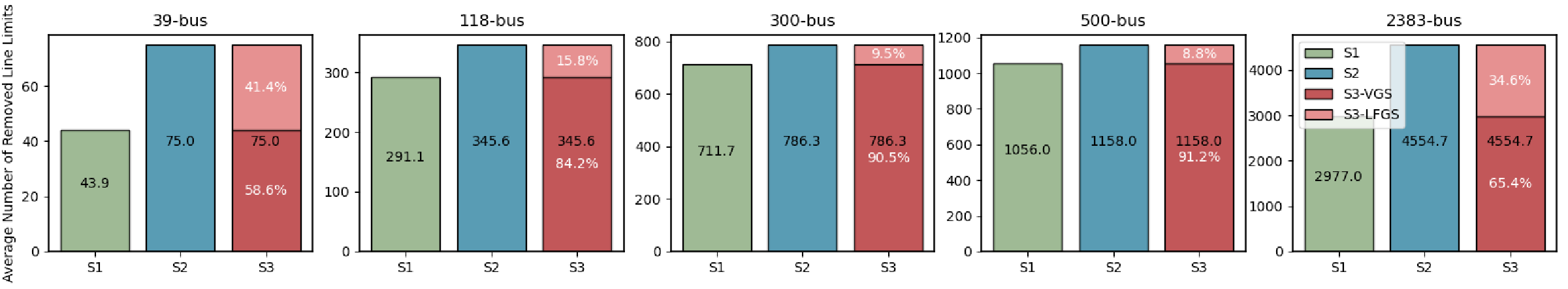}
 \vspace{-1em}
 \caption{The screening ability for redundant constraints of schemes S1-S3.} \label{GE:s1-s3_number}
% \vspace{-1.5em}	
\end{figure*}

\begin{figure*}[b] 
    \vspace{-1em}
    % \hspace{1em}
	\centering
\centering
\includegraphics[scale=0.45]{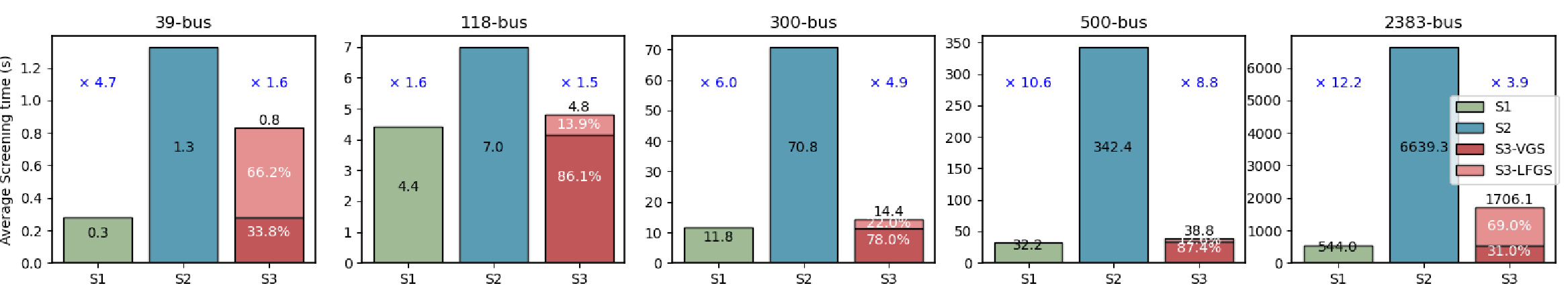}
 \vspace{-1em}
 \caption{The screening speed of schemes S1-S3. The proposed VGS achieves consistent speedup compared to the LFGS approach.}
	\label{GE:s1-s3_time}
\vspace{-1em}
\end{figure*}
Table \ref{VGS: Schemes} summarizes the screening schemes S1-S7 investigated here: \textbf{(1)} The schemes S1-S3 are applied to evaluate the screen time, the number of LPs solved, removed constraints, and time reduction in UC solutions of LFGS, VGS, and EOVL, respectively, verifying if LFGS removes more constraints, VGS screens faster, and EOVL combines their strengths. \textbf{(2)} The scheme S4 is simulated across different operating ranges to demonstrate Proposition \ref{operating_range}, verify screening applicability across varying load inputs, and examine the relationship between constraint-binding situations across ranges. \textbf{(3)} The schemes S5-S7 are implemented to evaluate and compare the improvements given by the cost cut, commitment cut, and their combination, demonstrating the Propositions \ref{VGS:cost} and \ref{VGS:unit}. To realize these evaluations and comparisons, we perform the case studies on the IEEE 39-, 118-, 300-, 500-, and 2383-bus power systems with the configurations given in Table \ref{GE: Configurations}. The load profiles are referred to as the corresponding cases in
MATPOWER \cite{zimmerman2010matpower}.

All simulations have been carried out on an unloaded MacBook Air with Apple M1 and 8G RAM. Specifically, all optimization problems are modeled and solved using YALMIP toolbox and MPT3 toolbox in MATLAB R2022b. 

% All neural networks for cost prediction are implemented by Python 3.9 with TensorFlow 2.10.0, Keras 2.10.0, and SKlearn 1.0.2.

\subsection{Performance Analysis for VGS, LFGS and EOVL}
\begin{table}[t]
\vspace{-1.5em}
% \hspace{0.1cm}
\centering
\caption{Testing System Configurations} \label{GE: Configurations}
%\vspace{-0.5em}
\setlength{\tabcolsep}{2.2mm}{
\begin{tabular}{cccccc}
\hline
System  & 39-bus & 118-bus & 300-bus & 500-bus & 2383-bus \\ \hline
No.Gen  & 10     & 54      & 69      & 90      & 200      \\
No.Line & 46     & 186     & 411     & 597     & 2896     \\ \hline
\end{tabular}}
\end{table}
The investigated LP-based screening has three basic backbones (VGS, LFGS, EOVL) as illustrated in Algorithm \ref{Algorithm: OACS}. To analyze the screening performance of these backbones, we test the schemes S1-S3 on five systems. From Fig. \ref{GE:s1-s3_number} and Fig. \ref{GE:s1-s3_time}, it can be seen that on the one hand, VGS removes fewer line limits for all systems, as the outer approximation in VGS causes a larger screening region than that of LFGS. On the other hand, VGS achieves the fastest screening speed, since it avoids solving the optimization model for each line limit. The larger the system scale is, the acceleration effect is more obvious, for the 2383-bus system, VGS can realize a 12.21x faster screening process compared to LFGS. 

To balance screening sufficiency and speed, we first apply VGS, followed by LFGS for the remaining line limits, which is the procedure of EOVL. Since the screening region of VGS encompasses that of LFGS, the constraints removed by VGS form a subset of those removed by LFGS method. Screening for the remaining constraints is then performed within LFGS screening region. Consequently, EOVL is expected to identify the same number of removed constraints as LFGS. The results confirm this, with VGS eliminating between 58.51\% and 91.22\% of the total constraints. Meanwhile, as mentioned in Section III-C, since the condition that $\Delta t_{V}\approx \Delta t_{LF}$ and $|N_{V}| > 2|G|$ holds, the screening times for all systems are sufficiently reduced, especially for the 500-bus system, EOVL can achieve an 8.81x acceleration. 
\begin{figure}[t] 
\vspace{-1.5em}
    % \hspace{0.35cm}
	\centering
	\centering
\includegraphics[scale=0.4]{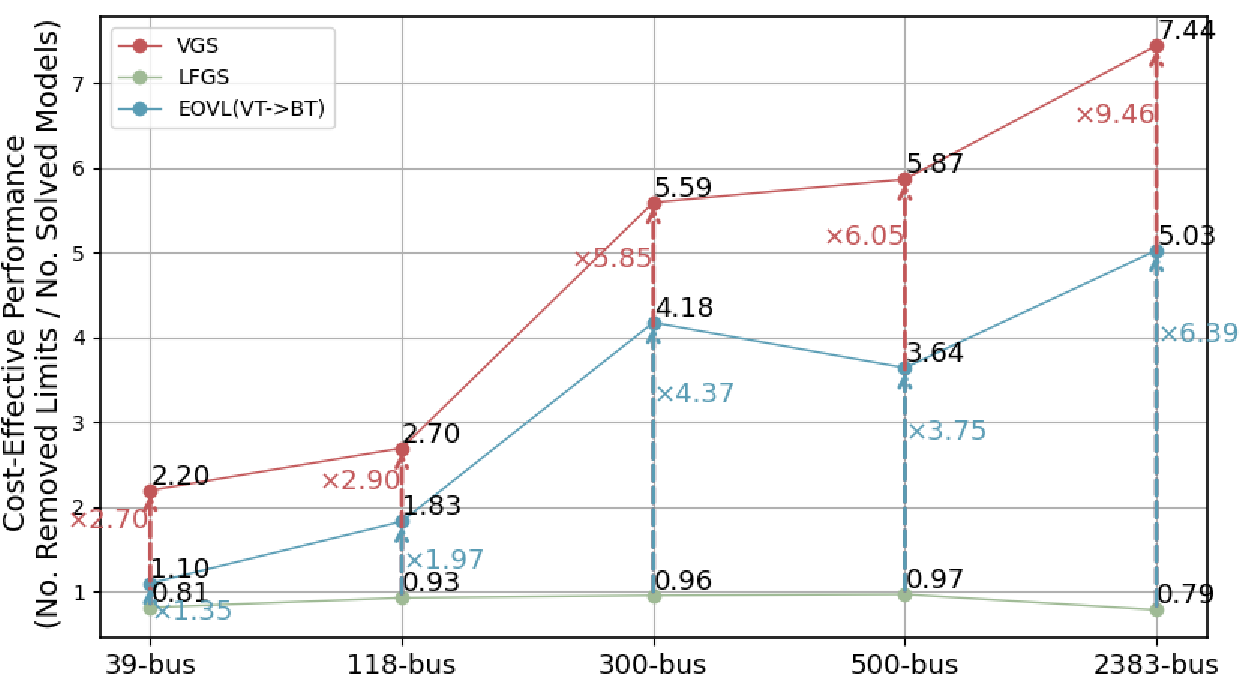}
    % \vspace{-0.5em}
	\caption{Comparison on computationally cost-effective performance. } \label{GE:cost-effective}
% \vspace{-1.5em}
\end{figure}
\begin{figure*}[ht] \vspace{-1em}
    % \hspace{1em}
	\centering
\centering
\includegraphics[scale=1.0]{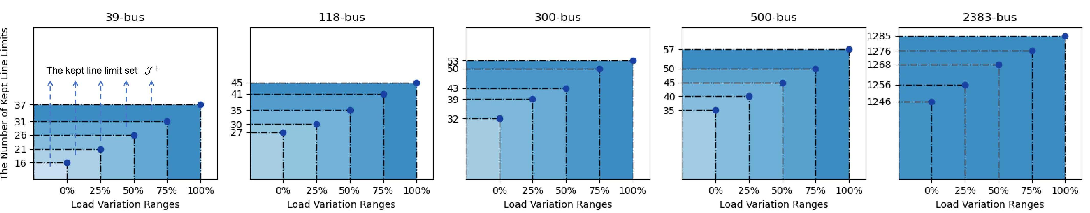}
 % \vspace{-2em}
 \caption{The screening results of S4 on different load operating ranges.}
	\label{GE:range}
 % \vspace{-1em}
\end{figure*}
The acceleration in VGS and EOVL demonstrates a higher computationally cost-effective performance, which can be defined as follows,
\begin{align} \label{screen_cost-effective}
r = \frac{N_{limits}}{N_{models}}.
\end{align}

In such a definition, $N_{limits}$ denotes the number of the removed constraints and $N_{models}$ is the number of LP models to solve. Fig. \ref{GE:cost-effective} shows that $r>1$ holds for both VGS and EOVL, achieving improvements of up to 9.46x compared to LFGS. This improvement is attributed to the fact that, in power systems, the number of units is generally much smaller than the number of transmission lines.

By removing redundant line limits from the UC models, the average solution time for them can be reduced from 12.44\% to 80.21\% without introducing a solution gap, as shown in Table \ref{GE: Time Reduction}. The difference in the results between S1 and S2/S3 demonstrates that UC solution time is closely linked to the number of removed limits and verifies the benefit of EOVL in speeding up both the UC and screening procedures.

\begin{figure*}[b] 
\vspace{-1em}
    % \hspace{-1.5em}
	\centering
\centering
\includegraphics[scale=0.5]{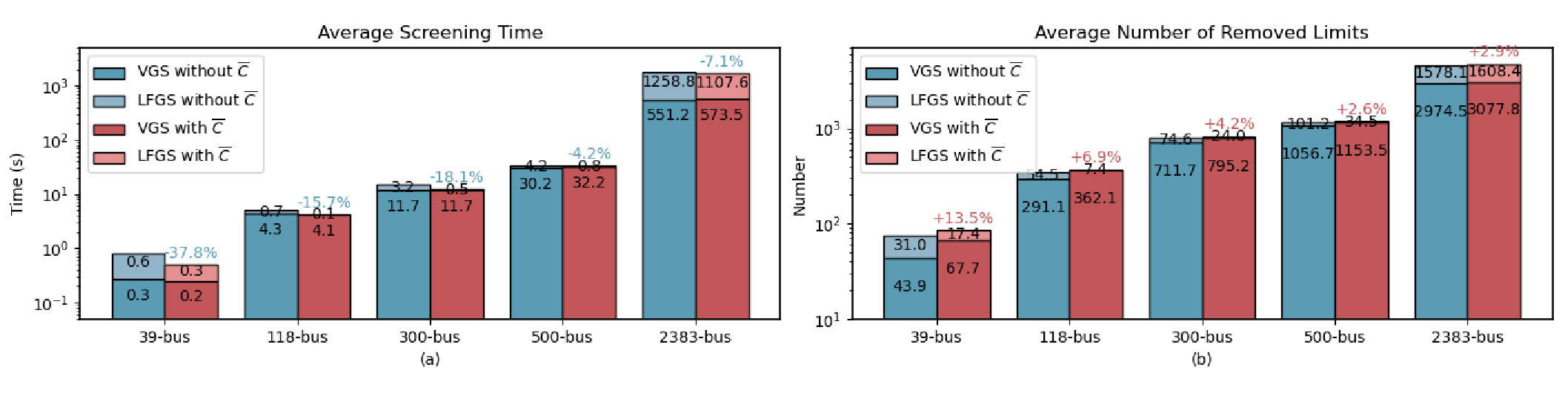}
 % \vspace{-3em}
 \caption{The screening performance of cost-cut integration (S5).}
	\label{GE:cost_bound_fig}
% \vspace{-1em}
\end{figure*}
\begin{table}[t]
% \vspace{-1em}
% \hspace{0.1cm}
\centering
\caption{Comparison on Solution Time Reduction $\Delta T$} \label{GE: Time Reduction}
% \vspace{-0.5em}
\setlength{\tabcolsep}{2.6mm}{
\begin{tabular}{cccccc}
\hline
 Systems     & 39      & 118     & 300     & 500     & 2383    \\ \hline
S1    & 2.73\%  & 8.43\%  & 22.53\% & 34.54\% & 31.82\% \\ 
S2/S3 & 12.44\% & 21.65\% & 36.36\% & 45.47\% & 80.21\% \\ \hline
\end{tabular}}
\vspace{-0.5em}
\end{table}
\begin{table}[t]
\vspace{-1.5em}
\centering
\caption{Solution Performance on UC instances with screened UC model} \label{GE: operating ranges}
\setlength{\tabcolsep}{1.5mm}{
\begin{tabular}{cccccc|c}
\hline
Systems                    & 39 & 118 & 300 & 500 & 2383 & Solu. Gap \\ \hline
$\Delta T$($\beta = 0.5$) & 2.03x & 2.13x   & 2.13x    & 2.06x  &  7.24x    &         0     \\
$\Delta T$($\beta = 1.0$) & 1.98x & 1.91x   & 2.15x    & 2.11x  &  6.89x    &         0     \\ \hline
\end{tabular}
}
\end{table}
\begin{table}[t]
\vspace{-5pt}
\centering
\caption{NN-based Cost Models Setting} \label{GE: NN_Configurations}
\setlength{\tabcolsep}{1.2mm}{
\begin{tabular}{cccccc}
\hline
Systems             & \multicolumn{3}{c}{Hidden layers} & Output layer & Activation units                           \\ \hline
    39/118/300/500/2383 & 100        & 50        & 30       &   1           & ReLU \\ \hline
\end{tabular}}
\vspace{-5pt}
\end{table}

\begin{table}[b]
\vspace{-2em}
\centering
\caption{Cost Bound Estimation Performance} \label{GE: NN_gap}
\setlength{\tabcolsep}{2mm}{
\begin{tabular}{ccccc}
\hline
Systems    & min($\frac{\hat{C}-C^{*}}{C^{*}}$)   & $\epsilon$     &   max($\frac{\overline{C}-C^{*}}{C^{*}}$)               & Solu. Gap              \\ \hline
39/118/300 & (-0.005, -0.001) & 0.005 & \multirow{2}{*}{$<1.5\%$} & \multirow{2}{*}{$0\%$} \\
500/2383   & (-0.01, -0.005) & 0.01  &                   &                          \\ \hline
\end{tabular}}
\vspace{-5pt}
\end{table}
\subsection{Improvements from Load Operating Range}
To verify Proposition \ref{operating_range} that VGS is valid for a load region, we use a polyhedral operating range defined in \eqref{region_definition} as $\mathcal{L}$ and integrate it into EOVL, which is the scheme S4.

\begin{align} \label{region_definition}
(1-\beta)\Tilde{\boldsymbol{l}} \leq \boldsymbol{l} \leq (1+\beta)\Tilde{\boldsymbol{l}};
\end{align}
here $\Tilde{\boldsymbol{l}}$ is the given nominal load profile and $\beta$ is the variation range. We test different settings of $\beta=20\%-100\%$ as shown in Fig. \ref{GE:range}. It shows across all systems that a larger value of $\beta$ keeps more line limits, while the set of kept line limits for smaller $\beta$ is a subset of that for larger $\beta$. This indicates we can remove more constraints with a more precise description of $\mathcal{L}$. In larger systems, the binding of constraints depends on more factors and is less influenced by the load profile. As a result, the percentage of kept constraints remains similar across different $\beta$ values. This suggests that a broader operating range can be incorporated to make the reduced model valid for more load inputs without overly complicating the reduced model.

In practice, the reduced models can be used to solve upcoming UC instances within the specified operating range. Here, we simulate cases of $\beta = 50\%$ and $100\%$. Table \ref{GE: operating ranges} presents the results on optimality and solution time, showing that the reduced model for a given $\mathcal{L}$ produces identical solutions for all tested instances, while the solution process is accelerated by an average factor of 1.91x to 7.24x.
\subsection{Improvements from Cost Bound}

\begin{figure*}[b] 
% \vspace{-1.5em}
    % \hspace{-2.5em}
	\centering
\centering
\includegraphics[scale=0.5]{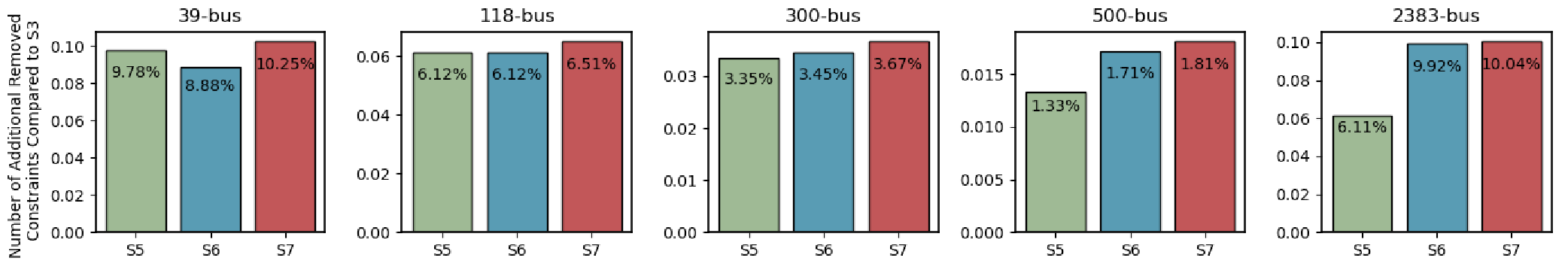}
 % \vspace{-2em}
 \caption{The screening ability of S5-S7.}
	\label{GE:prediction_fig_num}
\end{figure*}

\begin{figure*}[t] 
% \vspace{-1em}
    % \hspace{-2.5em}
	\centering
\centering
\includegraphics[scale=0.5]{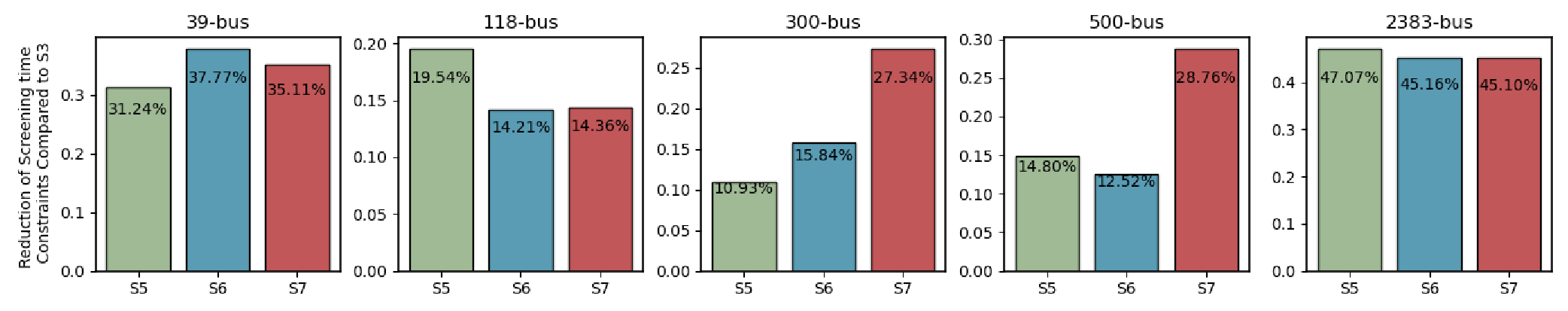}
 % \vspace{-2em}
 \caption{The screening time of S5-S7.}
	\label{GE:prediction_fig_time}
 % \vspace{-1em}
\end{figure*}

As mentioned in Section IV-C, we use the NN model and relaxation parameter $\epsilon$ to estimate the cost bound $\overline{C}=f_{NN}(\boldsymbol{\ell})(1+\epsilon)$. To train the NN model, we make use of the historical data of UC solutions along with the resulting total costs for supervised learning, with the loss function of minimizing the MSE error, i.e., 
% $L: =\left\|f_{NN}(\boldsymbol{\ell}) - J(\boldsymbol{\ell})\right\|_{2}^{2};$ 
\begin{align}  \label{Training: loss}
L: =\left\|f_{NN}(\boldsymbol{\ell}) - J(\boldsymbol{\ell})\right\|_{2}^{2};
\end{align}
where $J(\boldsymbol{\ell})$ denotes the real relationship between the load input and the output of UC cost. The neural networks are built on Tensorflow, with details explained in Table \ref{GE: NN_Configurations}. We collect 5,000 samples per neural network, with 20\% reserved for testing. To determine the relaxation parameter $\epsilon$, we look into the gap between the optimal UC costs $C^{*}$ and the predicted UC costs $\hat{C}$ on the 1,000 validation samples. As shown in Table \ref{GE: NN_gap}, for the 39-, 118- and 300-bus systems, $C^{*}$ is not less than 0.5\% of $\hat{C}$, while for 500-bus and 2383-bus systems, it is 1\%. Thus, we set $\epsilon = 0.005$ and $\epsilon = 0.01$ for the corresponding systems to ensure the feasibility of the cost cut \eqref{Screening: cost bound}. Next, we compare the gaps between the estimated bounds $\overline{C}$ and the optimal cost $C^{*}$, which are all less than $1.5\%$. This indicates $\overline{C}$ can provide sufficiently tight cuts to reduce the screening region while ensuring a 0\% solution gap after removing the additional constraints.

By adding the estimated cost cuts to the screening models, we obtain the screening results illustrated in Fig. \ref{GE:cost_bound_fig}. Compared to the cases without cost cuts, the screening times for all systems are reduced by 4.24\% to 37.81\%, while the total number of removed constraints increases by 2.60\% to 13.52\%, indicating that the computational cost-effectiveness of EOVL is further enhanced. Specifically, in the case of the 2383-bus system, both LFGS and VGS can remove more constraints with the cost cuts. In contrast, for the other systems, LFGS removes fewer constraints because a large portion of the constraints can be screened out in advance by VGS. As a result, the time reductions in the LFGS stage are more pronounced, since fewer constraints remain for it to check.
% Prob.$(\hat{C}<C^{*}): 2383: 989, 500: 965, 300:475, 118:129, 39:252
% Interval: 2383:-0.005386 500:-0.00640698, 300: -0.00466841, 118: -0.00129093, 39:-0.00103653

\subsection{Improvements from Commitment Cut}
Using a uniform distribution to obtain random loads within an operating range of $\beta=50\%$ and then solving \eqref{UC} for all generated loads to collect their UC solutions, we generate 5,000 load samples and record the on/off status i.e., the commitment in their UC solutions. Observing the generated samples, we first fix the status of units always on/off to 1/0 in the UC and screening models. Then, we apply the KNN method along with the generated dataset to predict the status of the remaining units for unseen 1,000 UC instances. We compare different settings of the KNN models and track their performance. The determined settings are presented in Table \ref{GE: KNN}, where we retain only the KNN predictors with 100\% validation accuracy and use their predictions for the corresponding units, leaving the on/off statuses of other units as decision variables. This helps reduce the screening complexity.

In this way, we implement the schemes S5-S7, and the results are given in Fig. \ref{GE:prediction_fig_num} and Fig. \ref{GE:prediction_fig_time}. It can be observed that for all systems, after fixing partial commitment, S6 can remove more constraints and screen faster. In most cases except the 39-bus system, the commitment cut can screen out more constraints than the cost cut. This suggests that, compared to the solutions satisfying the cost cut, those satisfying the commitment cut can generate more actual bounds—either line flow bounds or decision bounds—that are tighter than those obtained without any cuts. The results of S7 show that combining the cost cut with fixed unit statuses can further tighten the actual bounds and thus remove more constraints. Despite the increased complexity from additional constraints, screening time reductions exceed 10\% in all cases due to reduced variables and screened region, with the 2383-bus system achieving a reduction of up to 47.07\%. 

% \subsection{Large-Scale System Performance Analysis}
\begin{table}[tb]
\vspace{-1em}
\centering
\caption{Commitment Cut settings on predicted and fixed units.} \label{GE: KNN}
\setlength{\tabcolsep}{2.2mm}{\begin{tabular}{cccccc}
\hline
systems        & 39   & 118  & 300  & 500   & 2383 \\ \hline
K(No. Predicted units)              & 5(6) & 5(5) & 5(7) & 3(13) &  3(15)    \\
No.fixed units & 4    & 45   & 58   & 74    & 175  \\ \hline
\end{tabular}}
\vspace{-10pt}
\end{table}
\vspace{-1em}
\section{Conclusion} \label{Section_6}
% \vspace{-1em}
In this paper, we present a novel vertex-guided perspective on linear programming (LP)-based screening for unit commitment. Vertex-guided screening (VGS) that relies on the actual bounds of decision variables rather than line flow is proposed to avoid solving the LP model for each line limit. Through the ensemble of VGS and line flow-guided screening (LFGS), we balance the screening speed and sufficiency. The fruitful improvements of LP-based screening involving load operating ranges, cost cuts, and commitment cuts are integrated into the proposed VGS. In all investigated cases with different system configurations, our approach requires significantly less time to achieve the same performance as the classic method. In future work, we will explore suitable cost cuts and commitment cuts under different settings of load operating range.
\bibliographystyle{IEEEtran}
\bibliography{GC_IEEE}

\end{document}